\documentclass[draftcls,onecolumn,12pt]{IEEEtran}

\usepackage{amsmath,amstext,amsthm,amssymb,epsf}
\usepackage{latexsym}
\usepackage[pdftex]{graphicx}
\usepackage{setspace}
\usepackage{float}
\usepackage{placeins}
\numberwithin{equation}{section} % in amsmath

\newtheorem{theorem}{\sc Theorem}

\newtheorem{coro}{\sc Corollary}
\newtheorem{req}{\sc Requirement}

\newtheorem{defin}{\sc Definition}
\newtheorem{rem}{\sc Remark}
\newtheorem{cla}{\sc Claim}
\newtheorem{ex}{\sc Example}

\begin{document}

\title{How incomputable is Kolmogorov complexity?}
\author{
Paul M.B. Vit\'{a}nyi
\thanks{
Paul Vit\'{a}nyi is with the national research center for mathematics
and computer science in the Netherlands (CWI),
and the University of Amsterdam.
Address:
CWI, Science Park 123,
1098XG Amsterdam, The Netherlands.
Email: {\tt Paul.Vitanyi@cwi.nl}.
}}

\maketitle
\begin{abstract}
Kolmogorov complexity is the length of the ultimately compressed version of a file (that is, anything which can be put in a computer). Formally, it is the length of a shortest program from which the file can be reconstructed.
We discuss the incomputability of Kolmogorov complexity, which formal loopholes
this leaves us, recent approaches to compute or approximate 
Kolmogorov complexity, which approaches are problematic 
and which approaches are viable.

\end{abstract}

\section{Introduction}
Recently there have been several proposals how to compute or approximate in 
some fashion the Kolmogorov complexity function. 
There is a proposal that is popular as a reference in papers that
do not care about theoretical niceties, and a couple of proposals that do 
make sense but are not readily applicable. Therefore it is timely to survey 
the field and show what is and what is not proven. 

The plain Kolmogorov complexity was defined
in \cite{Ko65} and denoted by $C$ in the text \cite{LV19} and 
its earlier editions.
It deals with finite binary strings, {\em strings} for short.
Other finite objects
can be encoded into single strings in natural ways.
The following notions and notation may not be familiar to the
reader so we briefly discuss them.
The length of a string $x$ is denoted by $l(x)$.
The {\em empty string} of 0
bits is denoted by $\epsilon$. Thus $l(\epsilon)=0$.
Let $x$
be a natural number or finite binary string according to the
correspondence
\[
( \epsilon , 0),  (0,1),  (1,2), (00,3), (01,4), 
(10,5), (11,6), \ldots .
\]
Then $l(x)= \lfloor \log (x + 1) \rfloor$.
The Kolmogorov complexity $C(x)$ of $x$ is the length of a shortest
string $x^*$ such that $x$ can be computed from $x^*$ by a
fixed universal Turing machine (of a special type called ``optimal''
to exclude undesirable such machines). In this way $C(x)$ is
a definite natural number associated with $x$ and
a lower bound on the length of a compressed version of it by any
known or as yet unknown compression algorithm. We also use the
conditional version $C(x|y)$.

The papers by R.J. Solomonoff published in 1964, referenced as \cite{So64}, 
contain informal  
suggestions about the
incomputability of Kolmogorov complexity.
Says Kolmogorov, ``I came to similar conclusions [as Solomonoff],
before becoming aware of Solomonoff's work, in 1963--1964.''
In his 1965 paper \cite{Ko65} Kolmogorov
mentioned the incomputability of $C(x)$ without giving a proof:
``[$\ldots$] the function $C_{\phi} (x|  y)$ cannot be effectively calculated
(generally computable) even if it is known to be finite for
all $x$ and $y$.'' We give the formal proof of incomputability and 
discuss recent attempts to compute the Kolmogorov complexity partially, 
a popular but problematic proposal and some serious options. 
The problems of the popular proposal are discussed 
at length while the serious options are primarily restricted to brief 
citations explaining the methods gleaned from the introductions to the articles 
involved. 

\section{Incomputability}
To find the shortest program (or rather its length) for a string $x$ 
we can run all programs to see which one halts with output $x$ and select 
the shortest. We need to consider only programs of length at most that of $x$
plus a fixed constant. 
The problem with this process is known as the {\em halting problem}
 \cite{Tu36}: 
some programs do not halt and it is undecidable which ones they are.
A further complication is that we have to show there are infinitely many
such strings $x$ for which $C(x)$ is incomputable.

The first written proof of the incomputability of Kolmogorov complexity 
was perhaps in \cite{ZL70} and we reproduce it here following \cite{LV19} 
in order to show 
what is and what is not proved.
\begin{theorem}
\label{C5}
The function $C(x)$ is not computable.
Moreover, no partial computable function $\phi (x)$
defined on an infinite set of points can coincide
with $C(x)$ over the whole of its domain of
definition.
\end{theorem}
\begin{proof}
%This proof is related to that of Theorem~\ref{C4}, Item (iii).
We prove that there is no partial computable $\phi$
as in the statement of the theorem. Every infinite
computably enumerable set contains an infinite
computable subset, see e.g. \cite{LV19}.
Select an infinite computable subset $A$ in the
domain of definition of $\phi$.
The function
$\psi (m) = \min  \{ x: C(x)   \geq   m, x  \in  A  \}  $ is (total) computable
(since $C(x) = \phi (x)$ on $A$), and takes arbitrarily large values,
since it can obviously not be bounded for infinitely many $x$. 
Also, by
definition of $\psi$, we have $C( \psi (m))   \geq   m$. On the other hand,
$C( \psi (m))  \leq C_{\psi} ( \psi (m)) + c_{\psi}$ by definition of $C$,
and obviously $C_{\psi} ( \psi (m))  \leq l( m)$.
Hence, $m  \leq \log m$ up to
a constant independent of $m$, which is false from some
$m$ onward.
\end{proof}

That was the bad news; the good news is
that we can approximate
$C(x)$.
\begin{theorem}
\label{C6}
There is a total computable function $\phi (t, x)$,
monotonic decreasing in $t$, such that
$\lim_{{t}  \rightarrow   \infty} \phi (t, x) = C(x)$.
\end{theorem}
\begin{proof}
We define $\phi (t, x)$ as follows: For each $x$,
we know that the shortest program for $x$ has
length at most $l(x) + c$ with $c$ a constant independent of $x$.
Run the reference Turing machine $U$ (an optimal universal one)
for $t$ steps on %
\it each %
\rm program $p$ of length at most $l(x) +c$. If
for any such input $p$ the computation halts with
output $x$, then define the value of $\phi (t, x)$
as the length of the shortest such $p$,
otherwise equal to $l(x) + c$.
Clearly, $\phi (t, x)$
is computable, total, and monotonically nonincreasing
with $t$ (for all $x$, $\phi (t', x)  \leq \phi (t, x)$
if $t'   >   t$).
The limit exists, since for
each $x$ there exists a $t$ such that $U$
halts with output $x$ after computing $t$ steps starting
with input $p$ with $l(p) = C(x)$.
\end{proof}

One cannot decide, given $x$ and $t$, whether
$\phi (t, x) = C(x)$.
Since $\phi (t, x )$ is nondecreasing
and goes to the limit $C(x)$ for $t \rightarrow \infty$,
if there were a decision procedure to test
$\phi (t, x) = C(x)$, given $x$ and $t$, then we could compute
$C(x)$. But above we showed that $C$ is not computable.

However this computable approximation has no convergence guaranties as we
show now.
Let $g_1 , g_2 ,  \ldots $ be a sequence
of functions. We call $f$ the %
\it limit %
\rm of this sequence if $f(x) = \lim_{{t}  \rightarrow   \infty} g_t (x)$
for all $x$. The limit is %
\it computably uniform 
\rm if for every rational $\epsilon    >   0$ there exists a $t( \epsilon )$,
where $t$ is a total computable function,
such that $|f(x) - g_{t( \epsilon)} (x)|  \leq \epsilon$, for all $x$.
Let the sequence of one-argument functions $\psi_1 , \psi_2 , \ldots $
be defined by
$\psi_t (x) = \phi (t, x)$, for each $t$ for all $x$.
Clearly, $C$ is the limit
of the sequence of $\psi$'s. However, by Theorem~\ref{C5},
the limit is
not computably uniform. In fact, by the well-known halting problem,
for each $\epsilon   >   0$ and $t   >   0$
there exist infinitely many $x$ such that
$|C(x) - \psi_t (x)|   >   \epsilon$. This means that
for each $\epsilon   >   0$,
for each $t$ there are many $x$'s such that
our estimate $\phi (t, x)$ overestimates $C(x)$
by an %
\it error %
\rm of at least $\epsilon$.

\section{Computing the Kolmogorov complexity}
The incomputability of $C(x)$ does not mean that we can not compute $C(x)$
for some $x$'s. For example, if for individual string $x$ 
we have $C(C(x)|x)= c$ for some constant $c$, 
then this means that there
is an algorithm of $c$ bits which computes $C(x)$ from $x$. 
We can express the incomputability
of $C(x)$ in terms of $C(C(x)| x)$, which measures
what we may call the %
``complexity
of the complexity function.''
Let $l(x)=n$. It is easy to prove 
the {\it upper bound} $C(C(x)| x))  \leq \log n + O(1)$. But it is quite
difficult to prove the lower bound \cite{Ga74}:
For each length $n$ there are strings $x$ of length $n$ such that
$$
C(C(x)| x)   \geq   \log n - \log \log n - O(1)
$$
or its improvement by a game-based proof in \cite{BS14}: 
For each length $n$ there
are strings $x$ of length $n$ such that
$$
C(C(x)| x)   \geq   \log n - O(1).
$$
This means that $x$ only marginally helps to
compute $C(x)$; most information in $C(x)$
is extra information related to the halting problem. 

One way to go about
computing the Kolmogorov complexity for a few small values is as follows.
For example, let $T_1, T_2, \ldots$ be an acceptable
enumeration  of Turing machines. 
Such an acceptable enumeration is a formal concept \cite[Exercise 1.7.6]{LV19}. 
Suppose we have a fixed reference optimal universal
Turing machine $U$ in this enumeration.
Let $U(i,p)$ simulate $T_i(p)$ for all indexes $i$ 
and (binary) programs $p$.

Run $T_i(p)$ for all 
$i$ and $p$ in the following manner. As long as $i$ is sufficiently small 
it is likely that $T_i(p) < \infty$ for all $p$ 
(the machine $T_i$ halts for every $p$). 
The Busy Beaver function $BB(n): {\cal N} \rightarrow {\cal N}$ was 
introduced in \cite{Ra62} and has as value 
the maximal running time of $n$-state Turing machines in quadruple format 
(see \cite{Ra62} or \cite{LV19} for details).
This function is incomputable and rises faster than any computable function
of $n$.
 
Reference \cite{Br83} supplies the
maximal running time for halting machines for all $i < 5$ and for $i <5$ it 
is decidable which machines halt. For $i \geq 5$ 
but still small there are heuristics \cite{MB90,Mi04,Ke09,Ha16}. 
A gigantic lower bound for all $i$ is given in \cite{Gr64}.
Using Turing machines
and programs with outcome the target string $x$ we can determine an upper bound 
on $C(x)$ for reference machine $U$ (by for each $T_i$ encoding $i$ in 
self-delimiting format). Note that   
there exists no computable lower bound function approximating $C(x)$ since $C$ is incomputable and 
upper semicomputable. Therefore it can not be lower semicomputable \cite{LV19}. 

For an approximation using small Turing machines we do not have 
to consider all programs. If $I$ is the set of indexes of the Turing machines
and $P$ is the set of halting (or what we consider halting) programs then
\begin{align*}
\{(i,p)\}_x =  \{(i,p): T_i(p)  =x \}  \setminus \{(i,p): T_i(p)=x & \wedge 
\exists_{i',p'} (T_{i'}(p')=x \\ &  \wedge |(i,p)|  \leq  \min \{|i'|,|p'|\} \}, 
\end{align*}
with $i,i' \in I, p,p' \in P$. Here we can use the computably invertible
Cantor pairing function \cite{Wi} which is 
 $f: {\cal N} \times {\cal N} \rightarrow {\cal N}$
defined by $f (a,b) = \frac{1}{2} (a+b)(a+b+1) +b$ so that each 
pair of natural numbers $(a,b)$ is mapped to a natural number $f(a,b)$  
and vice versa. Since the Cantor pairing function is invertible, it must be one-to-one and onto: $|(a,b)|=|a|+|b|$.
Here $\{(i,p)\}_x$ is the desired set of applicable halting 
programs computing $x$. That is, if either $|i'|$ or $|p'|$ 
is greater than some $|(i,p)|$ with $(i,p) \in \{(i,p)\}_x$ while 
$T_{i'} (p')=x$ then we can discard  the pair concerned 
from $\{(i,p)\}_x$.

\section{Problematic Use of the Coding Theorem}

Fix an optimal universal prefix Turing machine $U$. 
The {\em Universal distribution} (with respect to $U$) is  
${\bf m}(x) = \sum 2^{-l(p)}$  where $p$ is a program (without input) for $U$
that halts. The prefix
complexity $K(x)$ 
is with respect to the same machine $U$. 
The complexity $K(x)$ is similar to $C(x)$ but such that the set of strings 
for which the Turing machine concerned halts is prefix-free (no program is 
a proper prefix of any other program). This leads to a slightly larger 
complexity: $K(x) \geq C(x)$.
The Coding theorem \cite{Le74}  states 
$K(x)= -\log {\bf m}(x)+O(1)$. 
Since $-\log {\bf m}(x) < K(x)$ (the term $2^{-K(x)}$ contributes to the sum 
and $2l(x)+O(\log x)$ is also a program for $x$)
we know that the $O(1)$ term is greater than 0. 

In \cite{Ze11} it was proposed to compute the Kolmogorov 
complexity by experimentally approximating the Universal distribution and
using the Coding theorem. This idea was used in several articles and 
applications.
One of the last 
is \cite{SZDG14}. It contains errors or inaccuracies for example: 
``the shortest program'' instead of ``a shortest program,'' 
``universal Turing machine'' instead of ``optimal universal Turing machine'' 
and so on. Explanation: there can be more than one 
shortest program, and Turing machines 
can be universal in many ways. For instance, if $U(p)=x$ for a universal
Turing machine, the Turing machine $U'$ such that $U'(qq)=U(q)$ for every 
$q$ and $U'(r)=0$ for every string $r \neq qq$ for some string $q$, is also
universal. Yet if $U$ serves to define the Kolmogorov complexity $C(x)$ then
$U'$ defines a complexity of $x$ equal to $2 C(x)$ which means that 
the invariance theorem does not hold for Universal Turing machines that are 
not optimal.

Let us assume that the computer used in the experiments fills 
the r\^ole of the required optimal Universal Turing machine for 
the desired Kolmogorov complexity, the target string, and the universal 
distribution involved. However, 
the $O(1)$ term in the Coding theorem is mentioned but otherwise ignored in the 
experiments and conclusions about the value of the Kolmogorov complexity 
as reported in \cite{Ze11,SZDG14}. Yet the experiments only concern 
small values of the Kolmogorov complexity, say smaller than 20, 
so they are likely swamped by the constant hidden in the $O(1)$ term. 
Let us expand on this issue briefly. In the proof of the Coding theorem, 
see e.g. \cite{LV19}, a Turing machine $T$ is used to decode a complicated code.
The machine $T$ is one of an acceptable enumeration $T_1, T_2, \ldots$ of all
Turing machines.
The target Kolmogorov complexity $K$ is shown to be smaller than the   
complexity $K_T$ associated with $T$ plus a 
constant $c$ representing 
the number of bits to represent $T$ and other items: $K(x) \leq K_T(x)+c$.
Since $T$ is complex since it serves to decode this code, 
the constant $c$ is huge, that is,
much larger than, say, $100$ bits.
The values of $x$ for which $K(x)$ is approximated by \cite{Ze11,SZDG14} 
are at most $5$ bits, that is, at most $32$. Unless there arises
a way to prove the Coding theorem without the large constant $c$, this method
is does not seem to work. Other problems: The distribution ${\bf m}(x)$ is 
apparently used as 
${\bf m}(x)= \sum_{i \in {\cal N}, T_i(\epsilon)=x} 2^{-l(\epsilon)}/i$,  
%\geq \sum_{i \in {\cal N}, T_i(\epsilon)=x} 2^{-l(\epsilon)}/i$ 
%with $j \leq i$  
see \cite[equation (6)]{STZ17}  
using  a  (noncomputable) enumeration of Turing machines 
$T_1,T_2, \ldots$ that halt on empty 
input $\epsilon$. 
Therefore $\sum_{x \in {\cal N}} {\bf m}(x)=
\sum_{i \in {\cal N}, T_i(\epsilon) < \infty} 2^{-l(\epsilon)}/i$ 
and with $l(\epsilon)=0$ we have $\sum_{x \in {\cal N}} {\bf m}(x) = \infty$
since  $\sum_{x \in {\cal N}} 1/x= \infty$.
By definition however  
$\sum_{x \in {\cal N}} {\bf m}(x) \leq 1$ : 
contradiction. 
It should be 
${\bf m}(x)= \sum_{i \in {\cal N}, T_i(p)=x} 2^{-l(p)-\alpha(i)}$ with 
$\sum_{i \in {\cal N}} \alpha(i) \leq 1$ as shown in 
\cite[pp. 270--271]{LV19}. 
%This means that the proof of the Coding theorem relying on the output
%of a reference optimal universal Turing machine  may not hold in the 
%setting of \cite{STZ17}. 

\section{Natural Data}\label{sect.nat}
The Kolmogorov
complexity of a file is a lower bound on the length
of the ultimate compressed version of that
file. We can approximate the Kolmogorov complexities involved by
a real-world compressor. Since the Kolmogorov complexity is incomputable,
in the approximation we never know how close we are to it. However,
we assume in \cite{CV04} that the natural data we are 
dealing with contain no complicated
 mathematical
constructs like $\pi=3.1415 \ldots$ or Universal Turing machines, 
see \cite{Vi13}.
In fact, we assume that the natural data we are dealing with contains
primarily effective regularities that a good compressor finds. Under those
assumptions the Kolmogorov complexity of the object is not much smaller than
the length of the compressed version of the object.

%In pattern recognition, learning, and data mining one obtains
%information from information-carrying objects.
%The notion of
%Kolmogorov complexity \cite{Ko65} is an objective measure
%for the information in an
%a {\em single} object, and information distance measures the information
%between a {\em pair} of objects \cite{BGLVZ98}.
%The information distance (in normalized form) can be used to to
%find the similarity between a pair of objects. We give a computable
%approximation using real-world compressors. This approximation
%is argued below to be in many cases sufficient in the case of natural
%data and good compressors.
%The resulting similarity measure is a parameter-free
%and alignment-free method \cite{Li02,CV04,CWV04,Ke04}. 
%\footnote{In bioinformatics the computation
%of the  similarity between genetic strings commonly involves the
%so-called ``alignment method.'' This method has
%high or even forbidding computational costs. See the remark in the example of
%phylogeny.
%For certain problems biologists look for alignment-free methods.}
%One of the advantages of this method is that it is very fast and also
%works for noisy objects \cite{CAO07}

\section{Safe Computations}
A formal analysis of the the intuitive idea in Section~\ref{sect.nat} 
was subsequently and independently given in \cite{BMRAA14}. 
From the abstract of \cite{{BMRAA14}}: 
``Kolmogorov complexity is an incomputable function. 
$\ldots$ By restricting the source 
of the data to a specific model class, we can construct a computable 
function to approximate it in a probabilistic sense: 
the probability that the error is greater than $k$ decays exponentially 
with $k$.'' This analysis is carried out but its application yielding concrete
model classes is not. 

\section{Short lists}
Quoting from \cite{Zi14}: ``Given that the Kolmogorov complexity is not 
computable, it is natural to ask
if given a string $x$ it is possible to construct a short list 
containing a minimal (plus possibly a small overhead) description of $x$. 
Bauwens, Mahklin, Vereshchagin and Zimand \cite{BMVZ13}
and Teutsch \cite{Te12} show that, surprisingly, the answer is YES. Even more,
in fact the short list can be computed in polynomial time. More precisely, the
first reference showed that one can effectively compute lists of quadratic size guaranteed to
contain a description of $x$ whose size is additively $O(1)$ from a minimal one (it is
also shown that it is impossible to have such lists shorter than quadratic), and
that one can compute in polynomial-time lists guaranteed to contain a description that is additively $O(\log n)$ from minimal. Finally, \cite{Te12} improved the latter
result by reducing $O(\log n)$ to $O(1)$''. See also \cite{TZ16}.

\section{Conclusion}
The review shows that the Kolmogorov complexity of a string 
is incomputable in general, but
maybe computable for some arguments. To compute or approximate the Kolmogorov 
complexity recently several approaches have been proposed. The most popular 
of these is inspired by L.A. Levin's Coding theorem and consists in 
taking the negative logarithm of the so-called universal
probability of the string to abtain the Kolmogorov complexity of very short 
strings (this is not excluded by incomputability as we saw). 
This probability is approximated 
by the frequency distributions obtained
from small Turing machines. As currently stated the  approach is problematic 
in the sense that it is only suggestive and can not be proved correct. 
Nonetheless, some applications 
make use of it. 
Proper approaches either restrict the domain of strings  of 
which the Kolmogorov complexity is desired (so that the incomputability 
turns into computability) or manage to restrict the Kolmogorov complexity 
of a string to an item in a small list of options (so that the Kolmogorov
complexity has a certain finite probability).

\bibliographystyle{plain}

\end{document}